\documentclass[10pt]{amsart}

\usepackage{xcomment,amssymb,latexsym,cite,upref,graphicx}

\usepackage[dvips]{color}
\usepackage[abs]{overpic}

\newtheorem{proposition}{Proposition}

\newcommand{\bR}{\mathbb{R}}

	\newcommand{\isogen}{\Lambda_{3}}  

\newcommand{\gf}{\mathfrak{g}}
\newcommand{\hf}{\mathfrak{h}}
\newcommand{\suf}{\mathfrak{su}}
\newcommand{\sof}{\mathfrak{so}}

\numberwithin{equation}{section}

\begin{document}

\title[Static Spherically Symmetric Solutions of the $SO(5)$ EYM equations]%
{Static Spherically Symmetric Solutions of the $SO(5)$ Einstein Yang-Mills Equations}

\author[R. Bartnik]{Robert Bartnik}
\address{School of Mathematical Sciences\\
Monash University, VIC 3800\\
Australia}
\email{robert.bartnik@sci.monash.edu.au}

\author[M. Fisher]{Mark Fisher}
\address{School of Mathematical Sciences\\
Monash University, VIC 3800\\
Australia}
\email{mark.fisher@sci.monash.edu.au}

\author[T.A. Oliynyk]{Todd A. Oliynyk}
\address{School of Mathematical Sciences\\
Monash University, VIC 3800\\
Australia}
\email{todd.oliynyk@sci.monash.edu.au}
\subjclass[2000]{83C20, 65L10}

\begin{abstract}
 Globally regular (ie.~asymptotically flat and regular interior),
 spherically symmetric and localised (``particle-like'') solutions of
 the coupled Einstein Yang-Mills (EYM) equations with gauge group
 $SU(2)$ have been known for more than 20 years, yet their properties are still not well understood.    
%
Spherically symmetric Yang--Mills fields are
 classified by a choice of isotropy generator and
%
%
 $SO(5)$ is distinguished as the simplest model with a
 \emph{non-Abelian} residual (little) group, $SU(2)\times U(1)$, and
 which admits globally regular particle-like solutions.
%
%
We exhibit an algebraic gauge condition which normalises the residual
gauge freedom to a finite number of discrete symmetries. This generalises the well-known reduction to the real magnetic potential
$w(r,t)$ in the original $SU(2)$ YM model. Reformulating
using gauge invariant polynomials dramatically simplifies the system
and makes numerical search techniques feasible. We find three families
of embedded $SU(2)$ EYM equations within the $SO(5)$ system, 
one of which was first detected only within the gauge-invariant
polynomial reduced system.  Numerical solutions representing mixtures of the three
$SU(2)$ sub-systems are found, classified by a pair of positive
integers.
	\end{abstract}

	\maketitle

\section{Introduction}

The LHC experiments currently underway at CERN are expected to settle
prominent vexing questions such as the origin of the Higgs mechanism,
the existence of supersymmetry in nature, and the meaning of dark
matter. However, regardless of the discoveries still to flow from the LHC, we can with
confidence predict that the governing equations will incorporate the
Einstein equations for the gravitational field and the Yang-Mills
equations for some non-Abelian gauge group.  From this perspective,
the spherically symmetric solutions of the $SU(2)$-EYM equations
discovered numerically in 
\cite{BMK}\footnote{Existence of global solutions was proved rigorously in \cite{Smetal1,Smetal2}.} 
are particularly important, for several
reasons.  First, the particle-like properties exhibited by the
solutions, namely static, asymptotically flat, and globally regular
with spatial topology $\bR^3$, confound previous expectations, based
on the known non-existence results for the vacuum Einstein equations
and for the YM equations separately, that such solutions could not
exist.  On the other hand, they confirm Wheeler's ``geon'' hypothesis
\cite{Wheeler}, that localised semi-bound gravitational (and with
hindsight, Yang--Mills) configurations might exist.
 	
Numerical and  perturbation  results \cite{Straumann,ChoptuikBizon} 
show that the $SU(2)$ EYM spherically symmetric  solutions may be 
viewed as an unstable balance between a dispersive YM  and the 
attractive gravitational force.  However, there are many aspects
of these systems which have not yet been studied in depth, and 
it seems premature to conclude that instability is inevitable --- see       the review \cite{Galtsov}.


The most detailed results have generally been obtained in the static
spherically symmetric setting with the simplest YM gauge group
$SU(2)$, which exhibits many of the basic properties of the general
non-Abelian group models. The spherically symmetric reduction of the
EYM equations can be viewed as a 2D EYM-Higgs system with a
\emph{residual gauge group}, $G^{\isogen}$, a residual Higgs field
(to be defined below), and a ``Mexican hat'' potential.  Most
attention has focused on the case where the gauge group is $SU(n)$ and
the residual group is Abelian, primarily because in this case it
is possible to completely fix the gauge freedom.  If the residual
group is non-Abelian, then it is known \cite{OK02b,OK03} that the
issue of gauge fixing becomes much more challenging.
	
To better understand gauge fixing  and solutions to the
the static spherically symmetric EYM equations we study $G=SO(5)$,
which is the simplest gauge group that supports a globally regular
model with \emph{non-Abelian} residual group. This model was discussed in
\cite{OK02b}, but the questions of gauge fixing and the existence of
solutions were not resolved.  We solve this problem here by using a 
related system satisfied by polynomials in the gauge fields which
are invariant under the action of the residual gauge group.

	\section{Static spherically symmetric field equations}
	
Spherical symmetry for Yang-Mills fields is complicated to define
because there are many ways to lift an $SO(3)$ action on space-time to
an action on the Yang-Mills connections. For real compact semisimple
gauge groups $G$, it was shown in \cite{B92,BS93} that equivalent
spherically symmetric Yang-Mills connections correspond to conjugacy
classes of homomorphisms of the isotropy subgroup, $U(1)$, into
$G$. The latter, in turn, are given by their generator $\isogen$,
the image of the basis vector $\tau_{3}$ of $\suf(2)$ (where
$\{\tau_{i}, i=1,\dots,3\}$ is a standard basis with
$[\tau_{i},\tau_{j}]=\epsilon^{k}_{\phantom{k}ij}\tau_{k}$), lying in
an integral lattice $I$ of a Cartan subalgebra $\hf_{0}$ of the Lie
algebra $\gf_{0}$ of $G$. This vector $\isogen$, when nontrivial,
then characterizes up to conjugacy an $\suf(2)$
subalgebra.\footnote{This ignores some interesting effects due to the
  fact that $SO(3)$ is not simply connected. For an analysis of
  $SO(3)$ actions on $SU(2)$ bundles see \cite{B97}.}
	
With the vector $\isogen\in \gf_0$ fixed, a gauge and a coordinate
system $(t,r,\theta,\phi)$ can be chosen \cite{B92,BS93} so that the
metric and gauge potential take the form
\begin{equation*}
ds^2=-S(r)^2\left(1-\frac{2m(r)}{r}\right)dt^2+\left(1-\frac{2m(r)}{r}\right)^{-1}dr^2+r^2\bigl(d\theta^2+\sin^2{\theta}\,
d\phi^2\bigr),
\end{equation*}
and\footnote{Here we are assuming the ansatz, that the gauge potential is ``purely
 magnetic''.}
\begin{equation*}
A =
\Lambda_1(r)d\theta+(\Lambda_2(r)\sin{\theta}+\isogen\cos{\theta})d\phi,
	\end{equation*}
	where $\Lambda_1(r),\Lambda_2(r)$ are $\gf_0$-valued maps
        satisfying
	\begin{equation} \label{Lcond}
	[\isogen,\Lambda_1(r)]=\Lambda_2(r), \quad \text{and} \quad
	[\Lambda_2(r),\isogen]=\Lambda_1(r).
	\end{equation}
	It is more convenient to use the following variables,
	\begin{equation} \label{Lpm}
	\Lambda_0 :=2i\isogen, \quad \text{and} \quad
	\Lambda_{\pm}:=\mp\Lambda_1-i\Lambda_2,
	\end{equation}
	which lie in $\gf$, the complexification of $\gf_0$. The conditions \eqref{Lcond} on
	$\Lambda_1$, $\Lambda_2$ then imply that the connection functions $\Lambda_{\pm}$ are valued in the residual Higgs bundles, $V_{\pm2}$, defined by
	\begin{equation} \label{V2def}
	\Lambda_{\pm}(r)\in V_{\pm 2}:=\left\{\,X\in\mathfrak{g}\, |\, [\Lambda_0,X]=\pm2 X\,\right\}.
	\end{equation}
	
	The \emph{residual (gauge) group}  $G^{\isogen}$ is defined as the connected Lie subgroup
	with Lie algebra
	\begin{equation*}
	\gf^{\isogen} := \{\, X\in \gf_0 \, | \, [\isogen{},X] = 0 \, \}.
	\end{equation*}
	We call the residual group (non-)Abelian if this is a (non-)Abelian Lie algebra. For $\gf^{\isogen}$ to be non-Abelian it is necessary but not sufficient that $\Lambda_0$ lies on the boundary of a Weyl chamber of $\hf$, a Cartan subalgebra of $\gf$.\\
	\\
	In the variables \eqref{Lpm}, the EYM equations become \cite{Oli02a}
	\begin{align}
	&m' =\frac{N}{2}||\Lambda_+'||^2+\frac{1}{8r^2}||\Lambda_0-[\Lambda_+,\Lambda_-]||^2, \label{meqn} \\
	&r^2\left(S^2N\Lambda_+'\right)'+S\left(\Lambda_+-\frac{1}{2}[[\Lambda_+,\Lambda_-],\Lambda_+]\right)=0, \label{ym2} \\
	&[\Lambda_+',\Lambda_-]+[\Lambda_-',\Lambda_+]=0, \label{ym1} \\
	&S^{-1}S'=\frac{1}{r}||\Lambda_+'||^2, \label{Seqn}
	\end{align}
	where $N=(1-2m/r)$. The $S$ parameter can be decoupled from the system, decreasing the order of the overall system by one but introducing a first order term in \eqref{ym2}. Also note that the second term in \eqref{ym2} is proportional to the gradient of the second term in \eqref{meqn}.  Requiring that the solutions are regular and asymptotically flat gives boundary conditions
	\begin{equation*}
	[\Lambda_+,\Lambda_-]=\Lambda_0, \quad \text{and} \quad
	[[\Lambda_+,\Lambda_-],\Lambda_+]=2\Lambda_+
	\end{equation*}
	at $r=0$, and as $r\rightarrow\infty$, respectively.
	
	The $||\cdot||$-norm is proportional, on each irreducible component of $\gf$, to the the real part of the Hermitean 
	inner-product derived from the Killing form \cite{Oli02a}.
	Multiplying $||\cdot||$ by a constant factor leads via a global rescaling of $m$ and $r$ to the original equations. Specifically, we have:
	
	\begin{proposition}\label{propscale}
	If $(m(r),\Lambda_+(r))$ satisfies \eqref{meqn}-\eqref{ym1}, then $(\alpha
        m(r/\alpha),\Lambda_+(r/\alpha))$ satisfies the equations obtained by replacing $||\cdot||$ in
        equations \eqref{meqn}-\eqref{ym1} with $\alpha||\cdot||$. 
	\end{proposition}
	\begin{proof}
	Substitution.
	\end{proof}
	The EYM system \eqref{meqn}-\eqref{ym1} typically admits subfamilies of solutions which also satisfy the original $SU(2)$ but with rescaled norm $||\cdot||$ (see Section 5).  We may regard such solutions are equivalent to the original $SU(2)$, since all values are consistently rescaled from the values in the original reports \cite{BFM,KSUN,KKSU3}.

	\section{An $SO(5)$ model}
	
	We now specialize to the gauge group $G=SO(5)$. The complexified Lie algebra $\gf=\sof(5,\mathbb{C})$ has the Cartan decomposition
	\begin{equation*}
	\mathfrak{g} = \mathfrak{h} \oplus \bigoplus_{\alpha \in R}\mathbb{C}e_{\alpha},
	\end{equation*}
	where $\mathfrak{h}=\text{span}_\mathbb{C}[h_{\alpha_1},h_{\alpha_2}]$ is the Cartan subalgebra, $R={\left\{\alpha_{\pm i},i=1,\ldots,4\right\}}$ is a root system in $\mathfrak{h}^*$
	with root diagram as in Figure \protect{\ref{rootdiagram}}.
	\setlength{\unitlength}{1mm}
	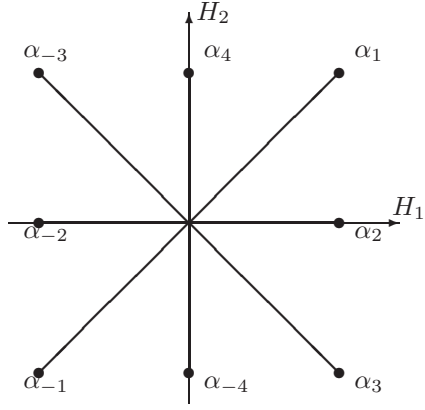
\begin{figure}[h]
	\begin{center}
	\begin{picture}(60,60)
	\put(10,10){\circle*{1.5}}
	\put(30,10){\circle*{1.5}}
	\put(50,10){\circle*{1.5}}
	\put(10,30){\circle*{1.5}}
	\put(50,30){\circle*{1.5}}
	\put(10,50){\circle*{1.5}}
	\put(30,50){\circle*{1.5}}
	\put(50,50){\circle*{1.5}}
	
	\thinlines
	\put(6,30){\vector(1,0){52}}
	\put(30,6){\vector(0,1){52}}
	
	\thicklines
	\put(30,30){\line(1,1){20}}
	\put(30,30){\line(-1,1){20}}
	\put(30,30){\line(1,0){20}}
	\put(30,30){\line(0,1){20}}
	\put(30,30){\line(-1,-1){20}}
	\put(30,30){\line(1,-1){20}}
	\put(30,30){\line(-1,0){20}}
	\put(30,30){\line(0,-1){20}}
	
	\put(57,31){$H_1$}
	\put(31,57){$H_{2}$}
	\put(8,8){$\alpha_{-1}$}
	\put(8,28){$\alpha_{-2}$}
	\put(8,52){$\alpha_{-3}$}
	\put(32,8){$\alpha_{-4}$}
	\put(32,52){$\alpha_{4}$}
	\put(52,8){$\alpha_{3}$}
	\put(52,28){$\alpha_2$}
	\put(52,52){$\alpha_1$}
	\end{picture}\caption{Root Diagram for $\sof(5,\mathbb{C})$}\protect{\label{rootdiagram}}
	\end{center}
	\end{figure}
	
	The reduced spherically symmetric EYM model depends strongly on the choice of isotropy generator $\Lambda_0$. For $G=SO(5)$, we choose \cite{OK02b}
	\begin{equation*}
	\Lambda_0 =2H_{1}.
	\end{equation*}
	With this choice, the residual group is $SU(2)\times U(1)$, while (see \eqref{V2def})
	\begin{equation*}
	V_2=\text{span}_{\mathbb{C}}[e_{\alpha_1},e_{\alpha_2},e_{\alpha_3}],
	\end{equation*}
	and hence
	\begin{align}
	\Lambda_+(r) &=w_1(r) e_{\alpha_1}+w_2(r)e_{\alpha_2}+w_3(r)e_{\alpha_3}, \label{so5Lp} \\
	\Lambda_-(r)&=\overline{w}_1(r)e_{\alpha_{-1}}+\overline{w}_2(r)e_{\alpha_{-2}}+\overline{w}_3(r)e_{\alpha_{-3}}. \label{so5Lm}
	\end{align}
	Substituting \eqref{so5Lp} and \eqref{so5Lm} into \eqref{meqn}-\eqref{ym1}, we find that the static spherically symmetric
	EYM equations in terms of the $w_i$ are
	\begin{align}
	m' = \frac{N}{2}(|w_1'|^2+|w_2'|^2+|w_3'|^2) + \frac{1}{2}{\mathcal{P}}&, \label{mso5eqn}\\
	r^2Nw_1''+(2m-\frac{1}{r}{\mathcal{P}})w_1'+w_1(1-|w_1|^2)+w_2(\frac{w_2\overline{w}_3}{2}-w_1\overline{w}_2)= 0&,
	\label{ym2so5.1}\\
	r^2Nw_2''+(2m-\frac{1}{r}{\mathcal{P}})w_2'+w_2(1-\frac{|w_2|^2}{2}-|w_1|^2-|w_3|^2)+\overline{w}_2w_1w_3= 0&,
	\label{ym2so5.2}\\
	r^2Nw_3''+(2m-\frac{1}{r}{\mathcal{P}})w_3'+w_3(1-|w_3|^2)+w_2(\frac{w_2\overline{w}_1}{2}-w_3\overline{w}_2)= 0&,
	\label{ym2so5.3}\\
	w_1'\overline{w}_1-w_1\overline{w}_1'+w_2'\overline{w}_2-w_2\overline{w}_2'+w_3'\overline{w}_3-w_3\overline{w}_3'=0&,\label{ym1so5.1}\\
	w_1'\overline{w}_1-w_1\overline{w}_1'-(w_3'\overline{w}_3-w_3\overline{w}_3')=0&,\label{ym1so5.2}\\
	w_1'\overline{w}_2-w_1\overline{w}_2'-(w_2'\overline{w}_3-w_2\overline{w}_3')=0&,\label{ym1so5.3}
	\intertext{where \begin{align}\nonumber
	{\mathcal{P}} &:= \frac{(1-|w_1|^2)^2}{2}+\frac{(1-|w_3|^2)^2}{2}+|w_2|^2(|w_1|^2+|w_3|^2+\frac{|w_2|^2}{4}-1)-\Re{[w_2^2\overline{w_1w_3}]}\\\nonumber
	&=\frac{1}{4}\left(|w_1|^2+|w_2|^2+|w_3|^2-2\right)^2+\frac{1}{4}\left(|w_1|^2-|w_3|^2\right)^2+\frac{1}{2}\left|w_1\overline{w}_2-w_2\overline{w}_3\right|^2 \end{align}}
	\intertext{The boundary conditions are}
	|w_1|^2+|w_2|^2+|w_3|^2=2&,\label{bcz.1}\\
	|w_1|-|w_3|=0&,\label{bcz.2}\\
	w_1\overline{w}_2-w_2\overline{w}_3=0&,\label{bcz.3}\\\nonumber
	\intertext{at $r=0$, and}
	w_1(1-|w_1|^2)+w_2(\frac{w_2\overline{w}_3}{2}-w_1\overline{w}_2)= 0&,\label{bci1}\\
	w_2(1-\frac{|w_2|^2}{2})-w_2(|w_1|^2+|w_3|^2)+\overline{w}_2w_1w_3= 0&,\label{bci2}\\
	w_3(1-|w_3|^2)+w_2(\frac{w_2\overline{w}_1}{2}-w_3\overline{w}_2)= 0&\label{bci3},
	\end{align}
	as $r\rightarrow\infty$.
	
	The sets of boundary conditions \eqref{bcz.1}-\eqref{bcz.3} and \eqref{bci1}-\eqref{bci3} define algebraic varieties on $\mathbb{C}^3$ and as such they do not restrict the initial and final values to one particular point in $\mathbb{C}^3$ (or indeed a discrete set of such points). Given a point lying on one of these varieties there will in general be a set of gauge-equivalent points but it is not immediately obvious what these equivalence classes of points are from the defining equations. This is a problematic feature of all models with a non-Abelian residual gauge group and will be discussed elsewhere.
	
	\section{Gauge fixing and the reduced equations}
	
	For the $SO(5)$ model, a direct computation shows that the ansatz
	\begin{equation}\label{Xuveq}
	\Lambda_+(r) = u(r)e_{\alpha_1}+0e_{\alpha_{2}}+v(r)e_{\alpha_{3}}, \, \, \, u(r),v(r)\in\mathbb{R}
	\end{equation}
	automatically satisfies the constraint equations \eqref{ym1} (equivalently 
	\eqref{ym1so5.1}-\eqref{ym1so5.3}), while the equations \eqref{meqn}-\eqref{ym2}
	(equivalently \eqref{mso5eqn}-\eqref{ym2so5.3}) reduce to
	\begin{align}
	&r^2Nu''+\left(2m-\frac{(1-u^2)^2+(1-v^2)^2}{2r}\right)u'+u(1-u^2)= 0, \label{red1}\\
	&r^2Nv''+\left(2m-\frac{(1-u^2)^2+(1-v^2)^2}{2r}\right)v'+v(1-v^2)= 0, \label{red2}\\
	&m' = \frac{N}{2}((u')^2+(v')^2) + \frac{1}{4r^2}({(1-u^2)^2}+{(1-v^2)^2), \label{red3}}
	\end{align}
	with boundary conditions
	\begin{equation} \label{redbc1}
	u^2+v^2=2, \quad |u|=|v|
	\end{equation}
	at $r=0$  and
	\begin{equation} \label{redbc2}
	u(1-u^2)= 0, \quad
	v(1-v^2)= 0
	\end{equation}
	as $r\rightarrow\infty$.
	
	The above results show that \eqref{Xuveq} is a consistent ansatz for the static spherically symmetric EYM equations. In fact the equations that result could have been obtained from a model with gauge group $SU(2)\times SU(2)$. 
	However, it is not obvious that the choice \eqref{Xuveq} is equivalent to fixing a section of the gauge fields.
	To see this we use the fact that the following two polynomials
	\begin{align*}
	K_1:=2||\Lambda_+||^2=&2\left(|w_1|^2+|w_2|^2+|w_3|^2\right),\\
	K_2:=2||[\Lambda_-,\Lambda_+]||^2=&4|w_1|^4+2|w_2|^4+4|w_3|^4\\
	+&8|w_1|^2|w_2|^2+8|w_3|^2|w_2|^2-8\mathrm{Re}[w_2^2\overline{w_1 w_3}],
	\end{align*}
	are a complete set of generators\footnote{This can be established by computing the Molien function to determine the appropriate orders of the polynomials and then verifying $K_1$ and $K_2$ are algebraically independent. See \cite{Forger} for some useful formulae.} for the ring of residual-group-invariant polynomials in $\Lambda_+$.  When $(w_1,w_2,w_3)=(u,0,v)$ for real $u,v$, these become
	\begin{equation*}
	K_1=2(u^2+v^2), \quad \text{and}\quad
	K_2=4(u^4+v^4).
	\end{equation*}
	These equations can be inverted to give
	\begin{equation*}
	u=\pm_1\left(\frac{K_1}{4}\pm_3\frac{\sqrt{2K_2-K_1^2}}{4}\right)^{\frac{1}{2}}, \quad \text{and} \quad
	v=\pm_2\left(\frac{K_1}{4}\mp_3\frac{\sqrt{2K_2-K_1^2}}{4}\right)^{\frac{1}{2}}
	\end{equation*}
	which are well-defined for all $K_1,K_2$ since the inequalities
	\begin{equation*}
	2K_2-K_1^2\geq 0, \quad \text{and} \quad
	K_2-K_1^2\leq 0
	\end{equation*}
	follow directly from the definition of $K_1$ and $K_2$. The $\pm$-signs amount to choosing which function is $u$ and which is $v$ and deciding on an arbitrary convention for their initial values.

	\section{Numerical results}
	
	From the form of the reduced equations \eqref{red1}-\eqref{red3}, it is clear that we can make two distinct simplifying
	assumptions $u=v$ or $v=\pm1$ (equivalently $u=\pm1$).
	
	\bigskip
	
	\noindent \textbf{(i)} Setting $u=v$, the reduced equations \eqref{red1}-\eqref{red3} become
	\begin{align*}
	&r^2Nu''+\left(2m-\frac{(1-u^2)^2}{r}\right)u'+u(1-u^2)= 0,\\
	&m' = N(u')^2 + \frac{1}{2r^2}(1-u^2)^2.
	\end{align*}
	which is the $SU(2)$ equation and has the well-known Bartnik-McKinnon family of solutions \cite{BMK,BFM}. 
	We note that these correspond to the embedded $SU(2)$ solutions as they satisfy 
	$\Lambda_{\pm}(r) = u(r)\Omega_{\pm}$ where $\Omega_{\pm}$ are fixed vectors satisfying 
	$[\Lambda_0,\Omega_{\pm}]=2\Omega_{\pm}$ and $[\Omega_+,\Omega_{-}]=\Lambda_0$. 
	
	\bigskip
	
	\noindent \textbf{(ii)} Setting $v=1$, the reduced equations \eqref{red1}-\eqref{red3} become
	\begin{align*}
	&r^2Nu''+\left(2m-\frac{(1-u^2)^2}{2r}\right)u'+u(1-u^2)= 0,\\
	&m' = \frac{N}{2}(u')^2 + \frac{1}{4r^2}(1-u^2)^2.
	\end{align*}
	As in Proposition \ref{propscale}, rescaling these equations by $\sqrt{2}$ in $m$ and $r$ results in the $SU(2)$-equations above, giving another family of solutions where the $u$ field is a radially-scaled Bartnik-McKinnon solution with smaller mass and the $v$ field remains constant.
	\begin{figure}
	\begin{overpic}[scale=.4,unit=1mm]%
	{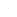}
	\put(0,0){\includegraphics[scale=.52]%
		{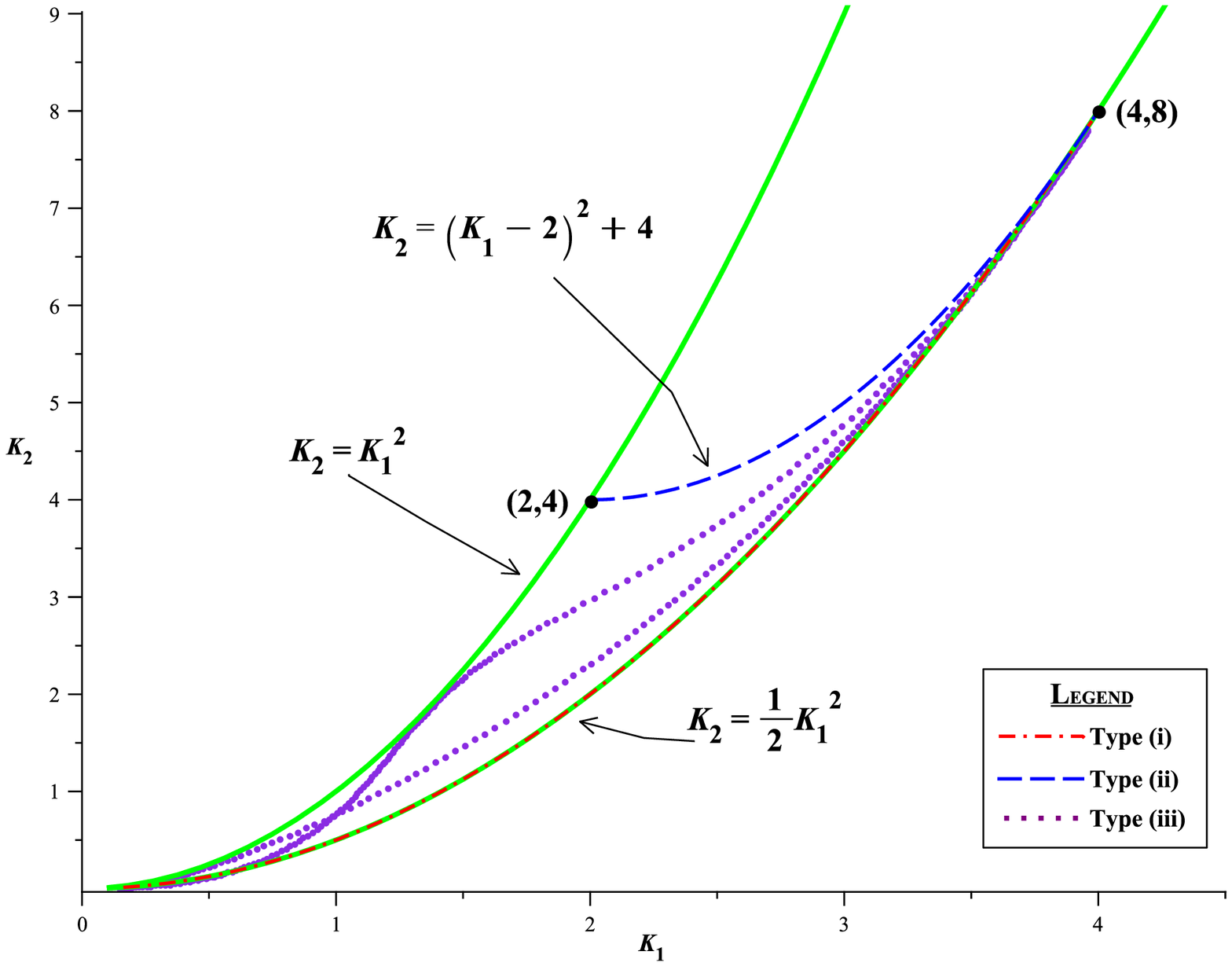}}
	\put(-34,14){\includegraphics[scale=.27]%
		{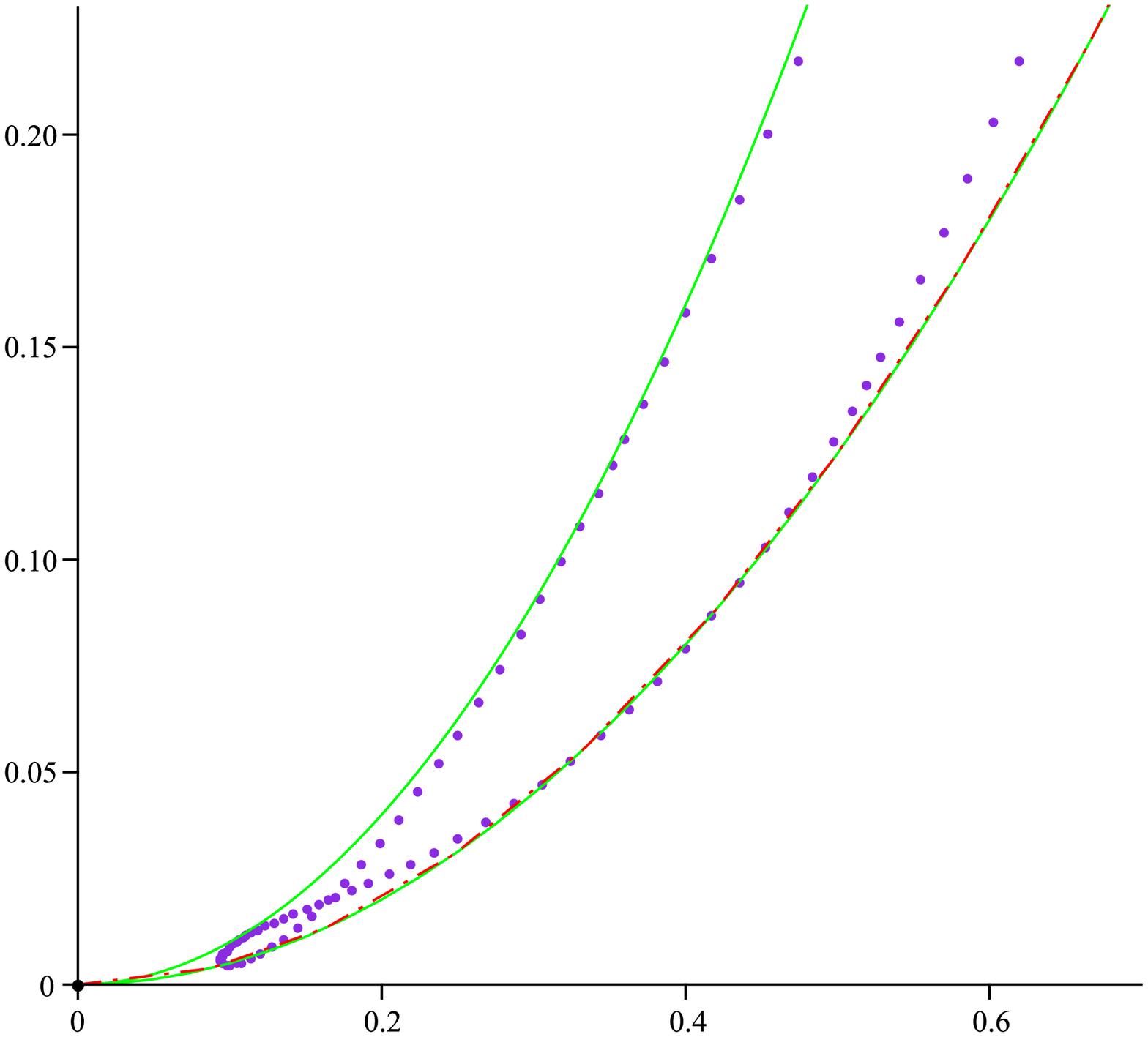}}
	\put(6,6){\framebox(16,3.5)}
	\put(-34,14.5){\framebox(52,48)}	
	\end{overpic}\caption{A plot of the solutions in the space of the invariant polynomials. All solutions start at $(4,8)$.The type (i) $SU(2)$ solutions lie on the $K_2=\frac{1}{2}K_1^2$ curve, while the type (ii) $SU(2)$ solutions lie on the $K_2=(K_1-2)^2+4$ curve. The type (iii) solutions are curves in the combined interior of the three parabolas touching the boundaries tangentially (an illustration of the (2,1) type (iii) solution is shown) }\label{parabFig}
	\end{figure}
	
	For the above two special situations, existence of a countably infinite number of solutions is guaranteed by the existence
	theorems in \cite{BFM,Smetal1,Smetal2}. Integrating the reduced equations \eqref{red1}-\eqref{red3} numerically, we 
	also find evidence for solutions that are not of the type (i) or (ii) above. We will call these solutions
	type (iii). Figure \ref{parabFig} indicates where these solutions lie in the space of invariant polynomials.
	All three types of solutions can be characterized by the number of nodes of the functions
	$u$ and $v$. Figure \ref{nodeFig} enumerates all of the solutions that were found numerically or
	that are known to exist analytically. The type (i) solutions lie on the diagonal $n_u=n_v$, while the
	type (ii) solutions lie on either the horizontal $n_v=0$ or vertical axis $n_u=0$. The remaining solutions
	are of type (iii) and are indicated on the node diagram by the larger circles. 
	We conjecture, based on the diagram, that for each $(n_u,n_v)\in \mathbb{Z}_{\geq0}\times
	\mathbb{Z}_{\geq 0}$ there exists a solution to the reduced equations \eqref{red1}-\eqref{red3} satisfying
	the boundary conditions \eqref{redbc1}-\eqref{redbc2}.
	\setlength{\unitlength}{1cm}
	\begin{figure}[ht]
	\begin{center}
	\begin{picture}(6,6)
	\put(0,0){\circle{0.1}}
	\put(0,1){\circle*{0.1}}
	\put(0,2){\circle*{0.1}}
	\put(0,3){\circle*{0.1}}
	\put(0,4){\circle*{0.1}}
	\put(0,5){\circle*{0.1}}
	\put(1,0){\circle*{0.1}}
	\put(2,0){\circle*{0.1}}
	\put(3,0){\circle*{0.1}}
	\put(4,0){\circle*{0.1}}
	\put(5,0){\circle*{0.1}}
	\put(1,1){\circle{0.1}}
	\put(1,2){\circle*{0.2}}
	\put(1,3){\circle*{0.2}}
	\put(2,1){\circle*{0.2}}
	\put(2,2){\circle{0.1}}
	\put(2,3){\circle*{0.2}}
	\put(3,1){\circle*{0.2}}
	\put(3,2){\circle*{0.2}}
	\put(3,3){\circle{0.1}}
	\put(4,4){\circle{0.1}}
	\put(5,5){\circle{0.1}}
	\put(1,4){-}
	\put(1,5){-}
	\put(2,4){-}
	\put(2,5){-}
	\put(3,4){-}
	\put(3,5){-}
	\put(4,5){-}
	\put(4,1){-}
	\put(5,1){-}
	\put(4,2){-}
	\put(5,2){-}
	\put(4,3){-}
	\put(5,3){-}
	\put(5,4){-}
	\put(0,-0.3){0}
	\put(1,-0.3){1}
	\put(2,-0.3){2}
	\put(3,-0.3){3}
	\put(4,-0.3){4}
	\put(5,-0.3){5}
	\put(-0.3,1){1}
	\put(-0.3,2){2}
	\put(-0.3,3){3}
	\put(-0.3,4){4}
	\put(-0.3,5){5}
	\thinlines
	\put(0,0){\vector(1,0){6}}
	\put(0,0){\vector(0,1){6}}
	
	\put(5.5,0.1){$n_u$}
	\put(0.1,5.5){$n_v$}
	\end{picture}
	\end{center}\caption{Node Diagram for solutions to the reduced equation } \label{nodeFig}
	\end{figure}
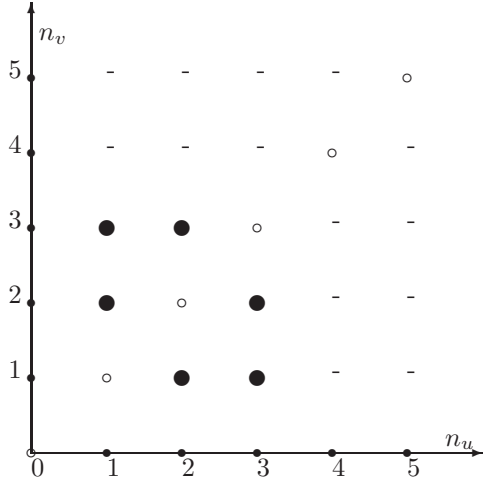

	For the numerical analysis we use the ``shooting to a fitting point" method \cite{NUMREC}, using the Matlab routine \texttt{ode45} which is a fourth-order Runge-Kutta solver. Using the power series expansions
	\begin{align*}
	&u=1+a_1r^2+(\frac{2}{5}a_1^3+\frac{2}{5}a_1a_2^2+\frac{3}{10}a_1^2)r^4 + O(r^6),\\
	&v=1+a_2r^2+(\frac{2}{5}a_2a_1^2+\frac{2}{5}a_2^3+\frac{3}{10}a_2^2)r^4 + O(r^6),\\
	&m=(a_1^2+a_2^2)r^3+(\frac{4}{5}a_1^3+\frac{4}{5}a_2^3)r^5 + O(r^7),
	\end{align*}
	near $r=0$, and
	\begin{align*}
	&u=\pm 1 + \frac{b_1}{r}+O(\frac{1}{r^2}),\\
	&v=\pm 1 + \frac{b_2}{r}+O(\frac{1}{r^2}),\\
	&m=m_{\infty}+O(\frac{1}{r^3}),
	\end{align*}
	near $r=\infty$. We shoot from the origin out trying to maximise (by searching in the parameters $a_i$ evaluated at $r=0.01$) the $r$-value that is attained before the solution violates one of the necessary conditions for a global smooth solution. Suitable large $r$-values will then give an indication what the $b_i$ and $m_{\infty}$ parameters should approximately be for the shooting back from some large $r$-value (we use $r=10000$). We then use these parameter values as an initial guess in the ``shooting to a fitting point" search, where the fitting was done at $r=10$.
	
	\section{Discussion}
	
	We have found three solutions using the numerical technique from the previous section. When combined with the two families of $SU(2)$ type solutions, they suggest that there is a solution for each point $(n_u,n_v)$ on the nonnegative integer lattice, where $n_u, n_v$ denote the number of nodes of $u$ and $v$ respectively. The numerical values are given in Table \ref{tab:TableOfNumericalParameters} and some type $(iii)$ solutions are displayed in Figures \ref{fig:fig21} and \ref{fig:fig32}.
	
	It is only with the reduced variables we have presented that these solutions could be found. As an indication of the difficulties otherwise, consider the particulars of the shooting method technique. In the original $w$-variables the ODE system is given in terms of four complex variables and the mass, $m$. This leads to eight real shooting parameters near $r=0$ and nine shooting parameters near $r=\infty$. Furthermore, after finding a promising approximate solution by shooting outward from $r=0$ and another by shooting inward from $r=\infty$, they will generally need to be gauge rotated to match at the fitting point (if indeed they are gauge equivalent-solutions). This all amounts to a difficult optimization problem in $17$ variables. This problem was discussed in \cite{OK02b} where some approximate solutions were obtained by shooting but a global solution was not found. 
	
	In contrast, shooting for the reduced variables has only five parameters to search in, two parameters near $r=0$ and three near $r=\infty$ with no need to gauge rotate possible solutions at the fitting point. The removal of degeneracy from the numerical problem and minimization of the number of shooting parameters is what has made it tractable to find the solutions presented here.
	
	\begin{table*}[ht]
		\centering
			\begin{tabular}{||c|c|c|l|l|l|l|c|c|c||}\hline
	$m_{\infty}$	&$u_{\infty}$	&$v_{\infty}$	  &$a_1$		&$a_2$ &$b_1$   &$b_2$		&$n_u$	  &$n_v$&type\\\hline
	0		 		&1	& 1	&0					&0					&0				&0					&0	&0&(i) \\
	0.58595	&-1	& 1	&-0.9074325	&0					&0.631715	&0					&1	&0&(ii) \\
	0.68685	& 1	& 1	&-1.303451	&0					&-6.26770	&0					&2	&0&(ii) \\		
	0.70380	&-1	& 1	&-1.394080	&0					&41.6713	&0					&3	&0&(ii) \\
	0.70655	& 1	& 1	&-1.409757	&0					&-259.038	&0					&4	&0&(ii) \\	
	0.70700	&-1	& 1	&-1.412337	&0					&1592.32	&0					&5	&0&(ii) \\
	0.70710	& 1	& 1	&-1.412759	&0					&-9770.35	&0					&6	&0&(ii) \\
	0.82865	&-1	&-1	&-0.45372		&-0.45372		&0.8934	  &0.8934	  	&1	&1&(i) \\
	0.92377	& 1	&-1	&-1.187117	&-0.117170	&-6.32743	&1.68255	 	&2	&1&(iii) \\
	0.93982	&-1	&-1	&-1.371735	&-0.022338	&37.63922	&1.67940	 	&3	&1&(iii) \\
	0.97135	& 1	& 1	&-0.65173		&-0.65173		&-8.8639	&-8.8639		&2	&2&(i) \\		
	0.98729	&-1	& 1	&-1.246544	&-0.155699	&42.52411	&-13.88445	&3	&2&(iii) \\
	0.99532	&-1	&-1	&-0.69704		&-0.69704		&58.9326	&58.9326		&3	&3&(i) \\
	0.99924	& 1	& 1	&-0.70488		&-0.70488		&-366.335	&-366.335		&4	&4&(i) \\
	0.99988	&-1	&-1	&-0.70617		&-0.70617		&2251.89	&2251.89		&5	&5&(i) \\
	0.99998	& 1	& 1	&-0.70638		&-0.70638		&-13817.4	&-13817.4		&6	&6&(i) \\\hline		
			\end{tabular}
		\caption{Table of Numerical Parameters}
		\label{tab:TableOfNumericalParameters}
	\end{table*}
	
	\begin{figure}
	\begin{overpic}[width=11.2cm,height=7cm]%
	{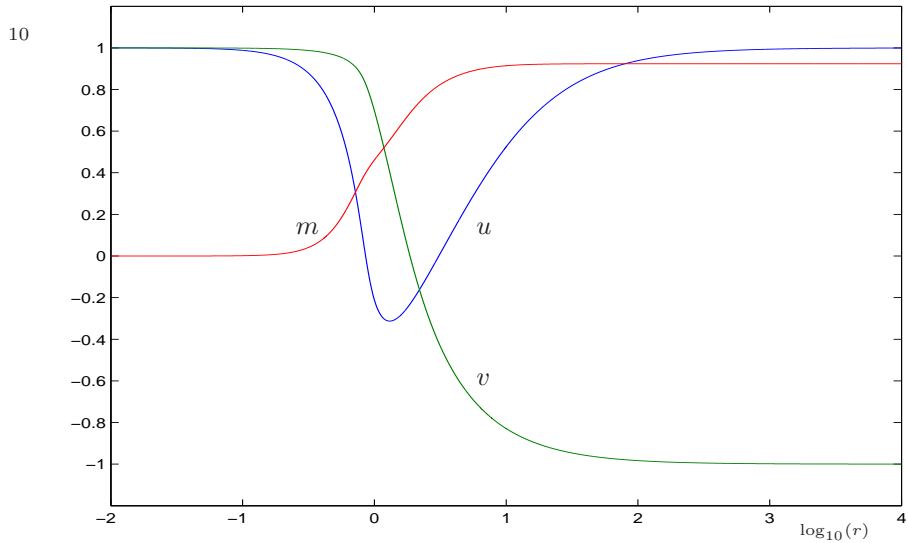}
	\put(3,3.9){$m$}
	\put(5.4,3.9){$u$}
	\put(5.4,1.9){$v$}
	\put(9.7,-0.1){{\tiny $\log_{10}(r)$}}
	\end{overpic}\caption{The $(2,1)$ type $(iii)$ solution. The mass, $m$, increases monotonically from $0$ to $0.924$, $u$ starts and ends at $1$ with two nodes and $v$ goes from $+1$ to $-1$ with one node.}\label{fig:fig21}
	\end{figure}
	
	\begin{figure}
	\begin{overpic}[width=11.2cm,height=7cm]%
	{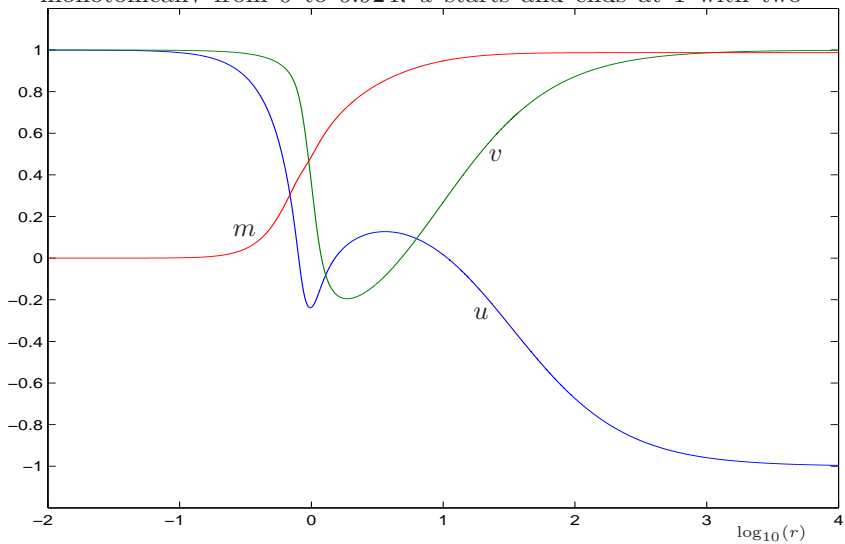}
	\put(3,3.9){$m$}
	\put(6.2,2.8){$u$}
	\put(6.4,4.9){$v$}
	\put(9.7,-0.1){{\tiny $\log_{10}(r)$}}
	\end{overpic}
	\caption{The $(3,2)$ type $(iii)$ solution. The mass, $m$, increases monotonically from $0$ to $0.987$, $u$ goes from $1$ to $-1$ with three nodes and $v$ goes from $+1$ to $+1$ with two nodes.}\label{fig:fig32}
	\end{figure}
	\eject

\end{document}